\documentclass[peerreview,draftclsnofoot,12pt,onecolumn]{IEEEtran}
\usepackage{graphicx}
\usepackage{amsmath}
\usepackage{epsfig}
\usepackage{cite}
\usepackage{caption}
\usepackage{subfigure}
\usepackage{amssymb}
\ifCLASSINFOpdf
\else
\fi
\begin{document}

\title{Duality of Channel Encoding and Decoding - Part II: Rate-1 Non-binary Convolutional Codes}

\author{Qimin You,
        Yonghui Li,~\IEEEmembership{Senior Member,~IEEE,}
        Soung Chang Liew,~\IEEEmembership{Fellow,~IEEE,}
        and Branka Vucetic,~\IEEEmembership{Fellow,~IEEE}
        
\thanks{Yonghui Li, Qimin You and Branka Vucetic are with School of Electrical and Information Engineering,
        University of Sydney, Sydney, NSW, 2006, Australia. Email: \{yonghui.li,~qimin.you,~branka.vucetic\}@sydney.edu.au}
\thanks{Soung C. Liew is with Department of Information Engineering, The Chinese University of Hong Kong, Hong Kong, Email:soung@ie.cuhk.edu.hk.}}
\maketitle

\begin{abstract}
This is the second part of a series of papers on a revisit to the bidirectional Bahl-Cocke-Jelinek-Raviv (BCJR) soft-in-soft-out (SISO) maximum a posteriori probability (MAP) decoding algorithm. Part I revisited the BCJR MAP decoding algorithm for rate-1 binary convolutional codes and proposed a linear complexity decoder using shift registers in the complex number field. Part II proposes a low complexity decoder for rate-1 non-binary convolutional codes  that achieves the same error performance as the bidirectional BCJR SISO MAP decoding algorithm. We observe an explicit relationship between the encoding and decoding of rate-1 convolutional codes in $GF(q)$. Based on this relationship, the BCJR forward and backward decoding are implemented by dual encoders using shift registers whose contents are vectors of complex numbers. The input to the dual encoders is the probability mass function (pmf) of the received symbols and the output of the dual encoders is the pmf of the information symbols. The bidirectional BCJR MAP decoding is implemented by linearly combining the shift register contents of the dual encoders for forward and backward decoding. The proposed decoder significantly reduces the computational complexity of the bidirectional BCJR MAP algorithm from exponential to linear with constraint length of convolutional codes. To further reduce complexity, fast Fourier transform (FFT) is applied. Mathematical proofs and simulation results are provided to validate our proposed decoder.
\end{abstract}

\newpage
\section{introduction}
In part I of this series of papers \cite{YQSB13,YRV12}, we revisited the bidirectional Bahl-Cocke-Jelinek-Raviv (BCJR) soft-in-soft-out (SISO) maximum a posteriori probability (MAP) decoding process of rate-1 binary convolutional codes. We observed an explicit relationship between the encoding and decoding of rate-1 binary convolutional codes and proposed a low complexity decoder using shift registers in the complex number field. The input to the decoder is the logarithm of soft symbol estimates of the coded symbols obtained from the received signals, and the output is the logarithm of the soft symbol estimates of the information symbols. The proposed decoder reduced the computational complexity of SISO MAP forward and backward recursion from exponential to linear without any performance loss.

The last few years have witnessed a drastic increase in the demand for reliable communications, constrained by the scarce available bandwidth, to support high-speed data transmission applications, such as voice, video, email and web browsing. To accommodate such demand, non-binary convolutional codes have been proposed to replace binary convolutional codes in many applications \cite{VY00, BG96, BM96}. For example, non-binary turbo codes, which employ non-binary convolutional codes as component codes, achieve lower error-floor and better performance at the waterfall region compared to binary turbo codes \cite{DF09}. Additionally, non-binary convolutional codes suit situations where bandwidth-efficient higher order (non-binary) modulation schemes are used, as well as situations where non-coherent modulation schemes are used, such as frequency-shift keying \cite{AT04}.

The main obstacle that impedes the practical implementation of the non-binary convolutional codes is the high decoding complexity. The decoding of non-binary convolutional codes is not equivalent to the decoding of binary convolutional codes, because the non-binary decoder operates on the symbol level, instead of bit level. Decoding is essentially finding an optimal path in a trellis based graph. Therefore, the Viterbi algorithm (VA) \cite{VAJ67} can be applied to decode non-binary convolutional codes. It provides the maximum-likelihood estimate of the information sequence based on the exhaustive search of the trellis over a fixed length window. Unfortunately, in standard VA, hard decision outputs are produced instead of soft outputs containing a posterior probability (APP) of the transmitted symbols. Therefore, the standard VA can not be used to decode concatenated codes, like turbo codes. To overcome this problem, the modified soft-output VA (SOVA) was proposed in \cite{LWY99} to decode non-binary convolutional codes. SOVA not only delivers the maximum-likelihood information sequence but also provides the APPs of the transmitted symbols. Therefore, it can be applied to decode concatenated codes. However, according to \cite{LC83}, the computational complexity of VA and SOVA is proportional to the number of states, which is $q^{K}$, where $q$ is the field size and $K$ is the constraint length \cite{WS87}. Thus, it grows rapidly for non-binary alphabets, which makes practical implementation tremendously difficult. Furthermore, both VA and SOVA suffer from considerable performance loss compared to the BCJR MAP algorithm \cite{BCJR74} which achieves optimal symbol error probability.

The BCJR MAP decoding algorithm is a bidirectional decoding algorithm which includes a forward and backward decoding recursion. The APP of the information symbol is estimated based on the combined forward and backward recursions. All the intermediate results during forward and backward recursions have to be stored before a decision is made, which incurs large memory storage requirements. Furthermore, the computational complexity at each time unit in both recursions is proportional to $q^{2K}$. The high computational complexity results in large decoding delay and unacceptably high costs.

Therefore, a low complexity decoder with good error performance is desirable for the pragmatic implementation of non-binary convolutional codes. In the decoding of turbo codes based on memory-1 convolutional codes in \cite{LSPC11}, the authors found that the encoder memory at the current time slot is a linear combination of the encoder memory at the previous time slot and the current input. Since the encoder memory at the previous time slot and the current input are independent, the probability mass function (pmf) of the current encoder memory can be calculated by the convolution of the pmf of encoder memory at the previous time slot and the pmf of the current input. This reduced the calculation complexity in the forward and backward recursions at each time slot to $q^2$. To further reduce complexity, fast Fourier transform (FFT) is employed on the pmf involved in the convolutions \cite{LSPC11}. The decoding complexity is thus reduced to $q\log_2 q$ at each time slot. However, this calculation simplification only works for memory-1 convolutional codes. The generalization to non-binary convolutional codes with arbitrary memory length is not considered in \cite{LSPC11}. Furthermore, forward and backward recursions still have to be performed based on the trellis of the non-binary convolutional codes. All the intermediate results have to be stored and thus large memory requirements are incurred.

In this paper, we propose a low complexity decoder for general rate-1 non-binary convolutional codes that achieves exactly the same error performance as the bidirectional BCJR MAP decoding algorithm. We observe an explicit relationship between the BCJR MAP forward/backward decoder of a convolutional code and its encoder. Based on this observation, we propose the dual encoders for SISO forward and backward decoding, which are simply implemented using shift registers whose contents are pmfs of complex vectors. Then, the bidirectional SISO MAP decoding is achieved by linearly combining the shift register contents in the forward and backward dual encoders. This significantly reduces the original exponential computational complexity of BCJR MAP forward and backward recursion to $(q^2)K$. To further reduce the computational complexity, FFT \cite{AAS96, MA07} is applied and its complexity is reduced to $(q\log_2 q)K$. Mathematical proofs and simulation results are provided to validate our proposed decoder.

The rest of this paper is organized as follows. In Section \ref{section forward dual encoder}, we propose an dual encoder for SISO MAP forward decoding of non-binary convolutional codes. The dual encoder for SISO MAP backward decoding is presented in Section \ref{section backward dual encoder}. The bidirectional BCJR MAP decoding is achieved by linearly combining the shift register contents of the forward and backward dual encoders, and simulation results are provided to validate our proposed decoder in Section \ref{shift register contents combining}. In Section \ref{conclusion}, concluding remarks are drawn. Mathematical proofs are given in the appendices.

%ity of encoding and forward decoding of
\section{Dual encoder of SISO MAP forward decoding}\label{section forward dual encoder}
In this section, we focus on the SISO MAP forward decoding algorithm. We consider convolutional codes over the finite fields with $q$ elements, denoted by $GF(q)$. We focus on the decoding of a single constituent rate-1 convolutional code in $GF(q)$, generated by $g(x)=\frac{a(x)}{f(x)}=\frac{1+a_1x+\cdots+a_{n-1}x^{n-1}+x^{n}}{1+f_1x+\cdots+f_{n-1}x^{n-1}+x^{n}}$, where $a_i, f_i \in GF(q), i=1, 2, \cdots, n-1$. Its encoder C is shown in Fig. \ref{Encoder of a rate-1 non-binary convolutional code}, where all the additions and multiplications are performed in $GF(q)$. In this paper, the input, output and memory of shift registers for convolutional encoders are in $GF(q)$. Let $ \overrightarrow{b}$$=\left(b_1, b_2, \cdots, b_L\right)$ and $\overrightarrow{c}$$=\left(c_1, c_2, \cdots, c_L\right)$ denote the information symbol sequence and the codeword sequence, where $L$ is the frame length. The code sequence is modulated and transmitted through the additive white Gaussian noise (AWGN) channel. The receiver obtains $\overrightarrow{ y_k}=\left(y_{1}, y_{2}, \cdots, y_{L}\right)$ at the output of the channel.

Let $p_{ c_k}(\omega)=p\left({c_k}=\omega|\overrightarrow{y_k}\right), \omega \in GF(q)$ denote the conditional probability of the code symbol ${ c_k}$ given $\overrightarrow{y_k}$. Let us further define the following probability mass function of ${c_k}$ and ${ b_k}$
\begin{align}
&{\bf P_{\bf c_k}}=[ p_{c_k}(0), p_{c_k}(1),  \cdots, p_{c_k}(q-1)]\\
&{\bf P_{\bf b_k}}=[ p_{b_k}(0), p_{b_k}(1),  \cdots, p_{b_k}(q-1)].
\end{align}

\vspace{5mm}
\begin{figure}[!ht]
   \centering
  \includegraphics[width=4.4in,height=2in]{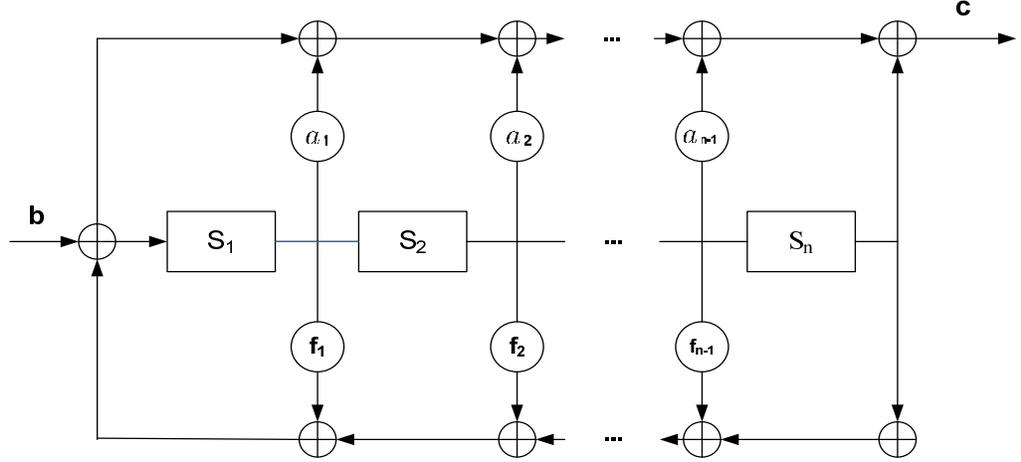}
  \caption{Encoder C of a rate-1 convolutional code in $GF(q)$, generated by $g(x)=\frac{1+a_1x+\cdots+a_{n-1}x^{n-1}+x^{n}}{1+f_1x+\cdots+f_{n-1}x^{n-1}+x^{n}}$, where all the additions and multiplications are performed in $GF(q)$.}\label{Encoder of a rate-1 non-binary convolutional code}
\end{figure}

\begin{figure}[!ht]
   \centering
  \includegraphics[width=2.0in,height=0.7in]{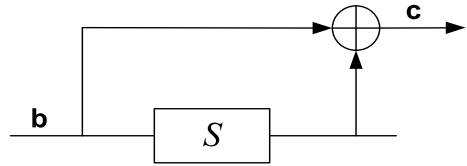}
  \caption{A convolutional code with generator polynomial $g(x)=1+x$ in $GF(4)$.}\label{example encoder C}
\end{figure}

\vspace{5mm}
\begin{figure}[!ht]
 \centering
  \includegraphics[width=2.0in,height=0.8in]{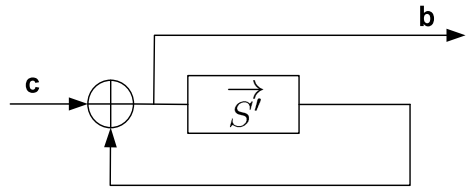}
  \caption{An encoder $\bar{C}$ with generator polynomial $q(x)=\frac{1}{1+x}$ in $GF(4)$.}\label{example decoder}
\end{figure}

The aim of the decoder is to derive ${\bf P_{\bf b_k}}$ based on the pmf of the code symbols. To facilitate the exposition, we first consider a simple example. Let us consider a convolutional code with generator polynomial $g(x)=1+x$ in $GF(4)$. Its encoder is shown in Fig. \ref{example encoder C}. We define an encoder $\bar{C}$ in $GF(4)$, described by $q(x)=\frac{1}{g(x)}=\frac{1}{1+x}$ (Fig. \ref{example decoder}). If the input to the encoder $\bar{C}$ is a codeword $\overrightarrow{ c}$, generated by $g(x)$, the output of the encoder $\bar{C}$ is the decoded information sequence $\overrightarrow{ b}$. Let $\overrightarrow{S'}(k)$ denote the memory of the shift register of encoder $\bar{C}$ at time $k$, the encoder output $b_k$ is given by
\begin{align}
b_k&=c_k+c_{k-1}\nonumber \\
&=c_k+\overrightarrow{S'}(k-1),
\end{align}
where $\overrightarrow{S'}(k)=c_k$. Therefore, we have the following relationship
\begin{align}\label{pbk}
p_{\bf b_k}(\omega)=p\left\{b_k=\omega\right\}=p\left\{c_k+\overrightarrow{S'}(k-1)=\omega\right\}.
\end{align}
Note that $c_k$ and $\overrightarrow{S'}(k-1)$ are independent, and equation (\ref{pbk}) can be written as
\begin{align}
p_{\bf b_k}(\omega)=\sum_{c_k=0}^{q-1}p\left\{c_k\right\}p\left\{\overrightarrow{S'}(k-1)=\omega-c_k\right\}.
\end{align}
According to the properties of random variables, since $b_k$ is the summation of $c_k$ and $\overrightarrow{S'}(k-1)$ in $GF(4)$, the pmf of $ b_k$ is the convolution of the pmf of $ c_k$ and $\overrightarrow{S'}(k-1)$. Let ${\bf P_{\overrightarrow{S'}(k-1)}}=[ p_{\overrightarrow{S'}(k-1)}(0), p_{\overrightarrow{S'}(k-1)}(1), \cdots, p_{\overrightarrow{S'}(k-1)}(q-1)]$ denote the pmf of $\overrightarrow{S'}(k-1)$, then the pmf of $b_k$ can be calculated as
\begin{align}
{\bf P_{\bf b_k}}={\bf P_{\bf c_k}} \ast {\bf P_{\bf \overrightarrow{S'}(k-1)}}.
\end{align}
where $\ast$ denotes the convolution operation. 

Its corresponding dual encoder is shown in Fig. \ref{example inverse decoder}. As verified mathematically in Appendix \ref{proof of theorem 1}, this dual encoder achieves exactly the same BER as the bidirectional BCJR MAP decoding algorithm.

%\vspace{5mm}
\begin{figure}[!ht]
 \centering
  \includegraphics[width=2.0in,height=0.8in]{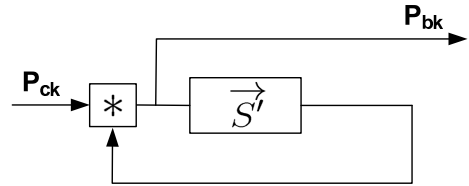}
  \caption{The dual encoder of the SISO forward decoding for the code $g(x)=1+x$ in $GF(4)$, where $*$ denotes the convolution operation.}\label{example inverse decoder}
\end{figure}

This dual encoder could be generalized to any rate-1 convolutional codes in $GF(q)$. First, we define an encoder $\bar{C}$ in $GF(q)$, described by $q(x)=\frac{1}{g(x)}=\frac{f(x)}{a(x)}=\frac{1+f_1x+\cdots+f_{n-1}x^{n-1}+x^{n}}{1+a_1x+\cdots+a_{n-1}x^{n-1}+x^{n}}$, shown in Fig. \ref{example encoder}. If the input to the encoder $\bar{C}$ is a codeword $\overrightarrow{c}$, generated by $g(x)$, the output of the encoder $\bar{C}$ is the decoded information sequence $\overrightarrow{b}$. In this encoding process, at each time instant, each encoder memory can be described as a linear combination of input symbols over $GF(q)$, denoted by $\overrightarrow{S}'_i(k)=\sum_{p=1}^k \eta_p c_p$, where $\eta_p \in GF(q)$. If the linear combination equations of two memories contain one or more common input symbols, we say that the two memories are correlated. For example, if $\overrightarrow{S}'_1(3)=c_1+c_3$, and $\overrightarrow{S}'_3(3)=c_1$, then these two memories are correlated, as both linear combination equations of these two memories contain the input symbol $c_1$.

\vspace{5mm}
\begin{figure}[!ht]
   \centering
  \includegraphics[width=4.4in,height=2in]{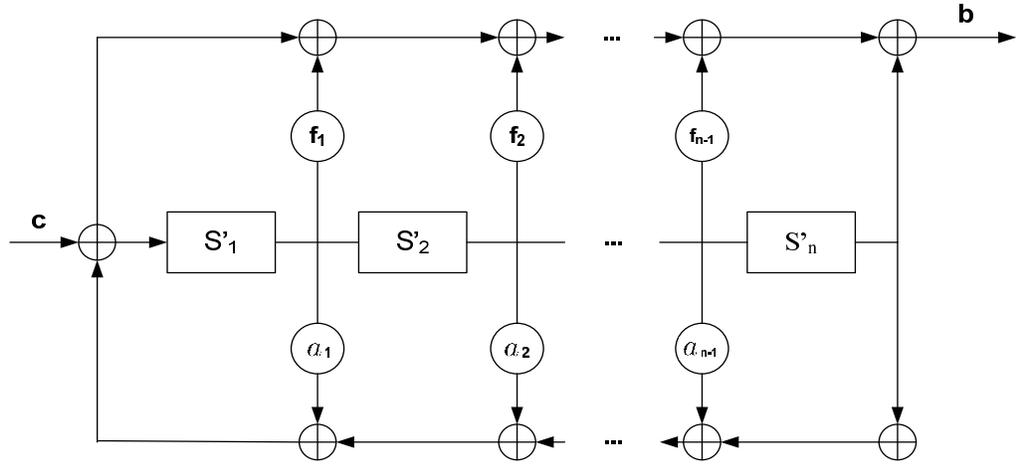}
  \caption{An encoder $\bar{C}$ in $GF(q)$, described by $q(x)=\frac{1+f_1x+\cdots+f_{n-1}x^{n-1}+x^{n}}{1+a_1x+\cdots+a_{n-1}x^{n-1}+x^{n}}$.}\label{example encoder}
\end{figure}

 If we map the structure in Fig. \ref{example encoder} to the convolutional structure in Fig. \ref{example convolutional encoder}, where each $+$ is replaced by $*$. The output of the dual encoder in Fig. \ref{example convolutional encoder} is different from the decoding output of the BCJR MAP forward decoding output, resulted from the correlation of the encoder memories of $\bar{C}$. Thus, when memories are correlated, the dual encoder cannot be used as an equivalent MAP forward decoding. To eliminate the memory correlation, we can multiply both the numerator and denominator of polynomial $q(x)$ by a common polynomial without actually changing the polynomial of $q(x)$. In order to obtain such a common polynomial, let us first define the minimum complementary polynomial for a given polynomial $a(x)$ as the polynomial of the smallest degree,
 \begin{align}
z(x)=1+z_1 x+ \cdots+z_{l-1}x^{l-1}+x^l,
\end{align}
such that
\begin{align}
a(x)z(x)=1+x^{n+l}.
\end{align}

 \vspace{5mm}
\begin{figure}[!ht]
   \centering
  \includegraphics[width=4.4in,height=2in]{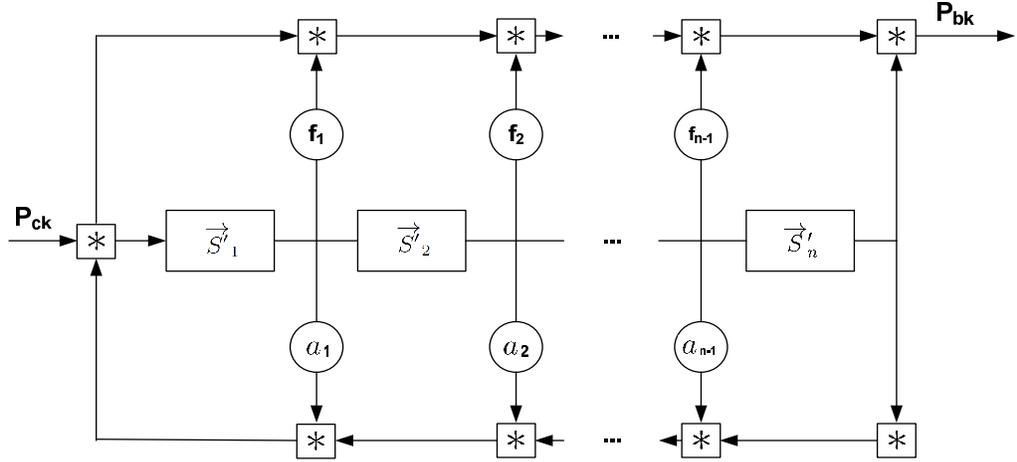}
  \caption{The dual encoder of the SISO forward decoding, described by $q(x)=\frac{1+f_1x+\cdots+f_{n-1}x^{n-1}+x^{n}}{1+a_1x+\cdots+a_{n-1}x^{n-1}+x^{n}}$.}\label{example convolutional encoder}
\end{figure}

Since $a(x)$ always divides $x^{q^n-1}+1$, the minimum complementary polynomial of $a(x)$ always exists. Let $f(x)z(x)=\left(1+f_1x+\cdots+f_{n-1}x^{n-1}+x^{n}\right)\left(1+z_1 x+ \cdots+z_{l-1}x^{l-1}+x^l\right)=1+h_1x+\cdots+h_{n+l-1}x^{n+l-1}+x^{n+l}$ and let $\overrightarrow{S'}_j(k), j=1, 2, \cdots, n+l,$ denote the memory of the $j$-th shift register of encoder $\bar{C}$, generated by $q(x)=\frac{f(x)z(x)}{a(x)z(x)}=\frac{1+h_1x+\cdots+h_{n+l-1}x^{n+l-1}+x^{n+l}}{1+x^{n+l}}=1+\frac{h_1x+\cdots+h_{n+l-1}x^{n+l-1}}{1+x^{n+l}}$. In encoder $\bar{C}$, the output is given by
\begin{align}\label{output at time k}
b_k=c_k+h_1\overrightarrow{S'}_1(k-1)+h_2\overrightarrow{S'}_2(k-1)+\cdots+h_{n+l-1}\overrightarrow{S'}_{n+l-1}(k-1),
\end{align}
and the memory of shift registers for encoder $\bar{C}$ can be expressed as
\begin{align}\label{s1k}
&\overrightarrow{S'}_1(k)=c_k+\overrightarrow{S'}_{n+l}(k-1)\\ \label{sjk}
&\overrightarrow{S'}_j(k)=\overrightarrow{S'}_{j-1}(k-1), j \geq 2.
\end{align}

Note that the additions and multiplications in the above three equations are performed in $GF(q)$. If we denote $h_i\overrightarrow{S'}_i(k-1), i=1, 2, \cdots, n+l-1$, as one symbol in $GF(q)$, then the value of $h_i\overrightarrow{S'}_i(k-1)$ equals $\overrightarrow{S'}_i(k-1)$ multiplied by $h_i$ in $GF(q)$. The pmf of $h_i\overrightarrow{S'}_i(k-1)$ is denoted by ${\bf P_{h_i\overrightarrow{S'}_i(k-1)}}$ and can be derived by cyclically shifting the $j$th element of $\bf P_{\overrightarrow{S'}_i(k-1)}$ to position $[jh_i]_q$. Let $\Pi_{h_i}$ denote such permutation of $\bf P_{\overrightarrow{S'}_i(k-1)}$ by $h_i$, where each $j$th element in $\bf P_{\overrightarrow{S'}_i(k-1)}$ is cyclically shifted to the $[jh_i]_q$ in $\Pi_{h_i}\bf P_{\overrightarrow{S'}_i(k-1)}$ \cite{GF04}. Based on (\ref{output at time k}), the probability that $b_k=\omega$ can be written as
\begin{align}\label{pbk=w}
p_{\bf b_k}(\omega)&=p\left\{b_k=\omega\right\}\nonumber\\
&=p\left\{c_k+h_1\overrightarrow{S'}_1(k-1)+h_2\overrightarrow{S'}_2(k-1)+\cdots+h_{n+l-1}\overrightarrow{S'}_{n+l-1}(k-1)=\omega\right\}.
\end{align}
Because $c_k$ and $h_1\overrightarrow{S'}_1(k-1)+h_2\overrightarrow{S'}_2(k-1)+\cdots+h_{n+l-1}\overrightarrow{S'}_{n+l-1}(k-1)$ are mutually independent, equation (\ref{pbk=w}) can be written as
\begin{align}\label{probability bk}
&p_{\bf b_k}(\omega)\nonumber\\
&=\sum_{c_k=0}^{q-1}p\left\{ c_k\right\} p\left\{h_1\overrightarrow{S'}_1(k-1)+h_2\overrightarrow{S'}_2(k-1)+\cdots+h_{n+l-1}\overrightarrow{S'}_{n+l-1}(k-1)=\omega - c_k\right\}
\end{align}
According to the definition of convolution, the probability mass function of ${b_k}$ can be expressed from (\ref{probability bk}) as
\begin{align}
{\bf P_{b_k}}={\bf P_{c_k}} \ast {\bf P_{h_1\overrightarrow{S'}_1(k-1)+h_2\overrightarrow{S'}_2(k-1)+\cdots+h_{n+l-1}\overrightarrow{S'}_{n+l-1}(k-1)}}.
\end{align}

Similarly, due to the independence of the memories of shift registers, the pmf vectors of $b_k$ can be represented as the convolution of the pmf vectors of $c_k$ and $h_j\overrightarrow{S'}_j(k)$
\begin{align}\label{input and output relationship}
{\bf P_{b_k}}&={\bf P_{c_k}} \ast {\bf P_{h_1\overrightarrow{S'}_1(k-1)}} \ast \cdots \ast {\bf P_{h_{n+l-1}\overrightarrow{S'}_{n+l-1}(k-1)}}\nonumber\\
&={\bf P_{c_k}} \ast \Pi_{h_1} {\bf P_{\overrightarrow{S'}_1(k-1)}} \ast \cdots \ast \Pi_{h_{n+l-1}} {\bf P_{\overrightarrow{S'}_{n+l-1}(k-1)}}.
\end{align}
Similarly, following (\ref{s1k}) and (\ref{sjk}), shift register contents of the dual encoder are updated as follows
\begin{align}\label{update of S1k of forward decoder}
&{\bf P_{\overrightarrow{S'}_1(k)}}={\bf P_{c_k}} \ast {\bf P_{\overrightarrow{S'}_{n+l}(k-1)}}\\ \label{update of shift register contents in the forward decoder}
&{\bf P_{\overrightarrow{S'}_j(k)}}={\bf P_{\overrightarrow{S'}_{j-1}(k-1)}}.
\end{align}

Based on the above analysis, we can derive a simple structure for MAP forward decoding implemented using the convolutional encoders, described by $q(x)=1+\frac{h_1x+\cdots+h_{n+l-1}x^{n+l-1}}{1+x^{n+l}}$, as shown in Fig. \ref{dual encoder for forward decoding.}, where $*$ denotes the convolution operation and $\Pi_{h_i}$ denotes permutation. The input of the dual encoder is ${\bf P_{c_k}}$ and the output of the dual encoder ${\bf P_{b_k}}$. Here the convolution operation is performed on complex vectors.

\vspace{6mm}
\begin{figure}[!ht]
  % \centering
  \includegraphics[width=4.3in,height=1.6in]{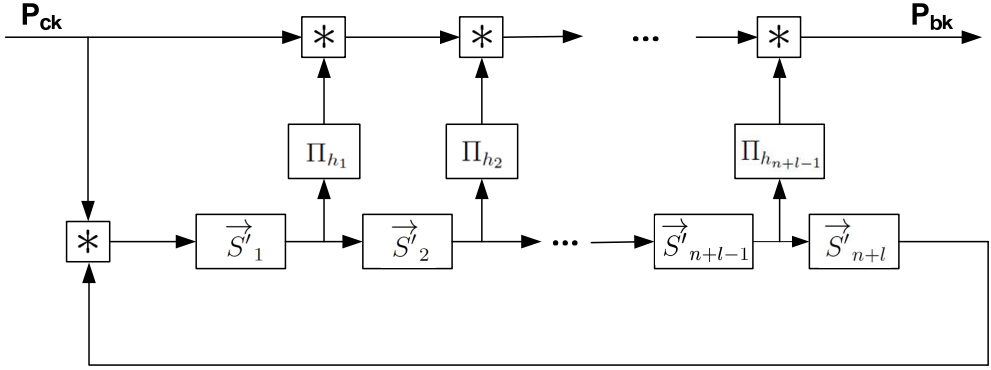}
  \caption{The dual encoder of the SISO forward decoding for the code $g(x)$, given by $q(x)=1+\frac{h_1x+\cdots+h_{n+l-1}x^{n+l-1}}{1+x^{n+l}}$.}\label{dual encoder for forward decoding.}
\end{figure}

Equations (\ref{input and output relationship}), (\ref{update of S1k of forward decoder}), (\ref{update of shift register contents in the forward decoder}) and Fig. \ref{dual encoder for forward decoding.} reveal an interesting relationship of the convolutional encoder and the SISO forward decoder for rate-1 convolutional codes in $GF(q)$. This can be summarized in the following theorem.

%A duality relationship of the encoder and soft input and soft output (SISO) forward decoder of a rate-1 non-binary convolutional code. This can be summarized in the following theorem.
\newtheorem{theorem}{\bf Theorem}
\begin{theorem}\label{theorem 1}
{\bf Dual encoder for SISO MAP forward decoding:} For a rate-1 convolutional code in $GF(q)$, generated by $g(x)=\frac{1+a_1x+\cdots+a_{n-1}x^{n-1}+x^{n}}{1+f_1x+\cdots+f_{n-1}x^{n-1}+x^{n}}$, we define its dual encoder as the encoder with inverse generator polynomial of $g(x)$, given by $q(x)=\frac{1}{g(x)}=\frac{f(x)z(x)}{a(x)z(x)}=1+\frac{h_1x+\cdots+h_{n+l-1}x^{n+l-1}}{1+x^{n+l}}$. Then the SISO MAP forward decoding of convolutional codes can be implemented by its dual encoder of complex vectors, which is shown in Fig. \ref{dual encoder for forward decoding.}. The output of the dual encoder is the pmf of the information sequence. Note that all the operations in the dual encoder are convolution operations.
\end{theorem}

\begin{proof}
See Appendix \ref{proof of theorem 1}.
\end{proof}

The complexity of the dual encoder for forward decoding in Fig. \ref{dual encoder for forward decoding.} is dominated by the convolution operations, and thus scales as $\mathcal{O}\left(q^2K\right)$ \cite{RG68}. The complexity can be further reduced by applying the FFT on the probability vectors involved in the convolutions \cite{DM99}. Let $F[{\bf P_1}]=\left(F[P_1](0), F[P_1](1), \cdots, F[P_1](q-1)\right)$ and $F[{\bf P_2}]=\left(F[P_2](0), F[P_2](1), \cdots, F[P_2](q-1)\right)$ be the FFT transformed vectors of ${\bf P_1}$ and ${\bf P_2}$, we define the Hadamard product, which is the element-wise multiplication of two vectors \cite{Robert12}, of $F[{\bf P_1}]$ and $F[{\bf P_2}]$ as $F[{\bf P_1}] \circ F[{\bf P_2}]=(F[P_1](0)F[P_2](0), F[P_1](1)F[P_2](1), \cdots, $$F[P_1](q-1)F[P_2](q-1))$. The Fourier transform of the convolution of two functions equals the product of the Fourier transforms of these two functions \cite{OSB09}. Therefore, (\ref{input and output relationship}), (\ref{update of S1k of forward decoder}) and (\ref{update of shift register contents in the forward decoder}) can be expressed as
\begin{align}
&{\bf P_{b_k}}=F^{-1}\left\{{F[\bf P_{c_k}}] \circ F[\Pi_{h_1} {\bf P_{\overrightarrow{S'}_1(k-1)}} ]\circ \cdots  \circ F[\Pi_{h_{n+l-1}} {\bf P_{\overrightarrow{S'}_{n+l-1}(k-1)}}]\right\}\\
&{\bf P_{\overrightarrow{S'}_1(k)}}=F^{-1}\left\{ F[{\bf P_{c_k}}] \circ F[ {\bf P_{\overrightarrow{S'}_{n+l}(k-1)}}]\right\}\\
&{\bf P_{\overrightarrow{S'}_j(k)}}={\bf P_{\overrightarrow{S'}_{j-1}(k-1)}}.
\end{align}
Therefore, we propose an FFT dual encoder for forward decoding, shown in Fig. \ref{forward dual encoder}, where $\circ$ denotes the element-wise multiplication of two vectors and $F$ the FFT of a vector. Note that all the convolution operations in the dual encoder in Fig. \ref{dual encoder for forward decoding.} become the element-wise multiplication in Fig. \ref{forward dual encoder}, and this considerably reduces the complexity from $\mathcal{O}\left(\left(q^2\right)K\right)$ to $\mathcal{O}\left(\left(q\log_2 q\right)K\right)$. Note that the output of the FFT dual encoder is exactly the same as the output of the dual encoder for forward decoding.

\vspace{8mm}
\begin{figure}[!h]
  % \centering
  \includegraphics[width=4.6in,height=1.9in]{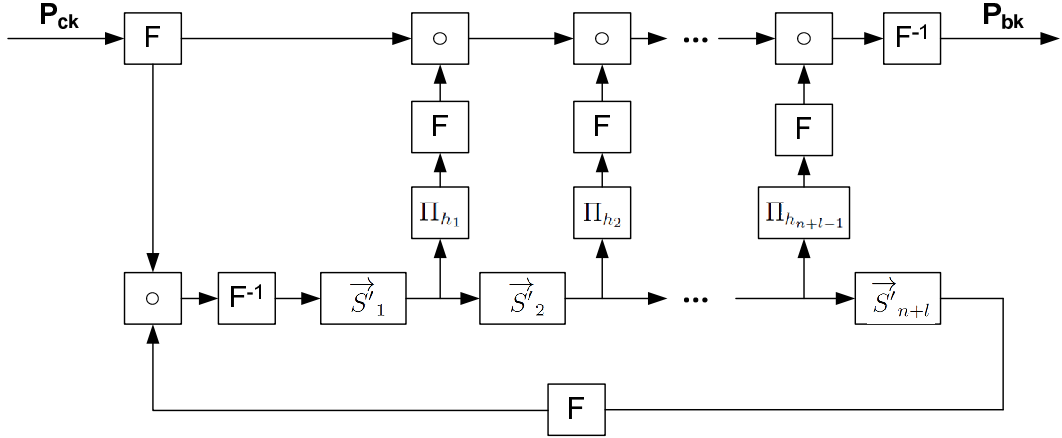}
  \caption{The FFT dual encoder of SISO forward decoding for the code $g(x)$.}\label{forward dual encoder}
\end{figure}

\section{Inversev encoder of SISO MAP backward decoding}\label{section backward dual encoder}
%reverse memory labeling
In this section, we propose an dual encoder for the BCJR MAP backward decoding of rate-1 convolutional codes in $GF(q)$. In the BCJR MAP backward decoding, the received signals are decoded in a time-reverse order. That is, given the received signal sequence ${\bf y}=\left(y_1, y_2, \cdots, y_{L}\right)$, the order of the signals to be decoded is from $y_{L}$, $y_{L-1},$ to $y_1$. In addition, in the backward decoding, the decoder has to follow the trellis in a reverse direction. Figs. \ref{forward representation} and \ref{trellis representation} show the encoder and its trellis, described by the generator polynomial $g(x)=\frac{1+2x}{1+x}$. Fig. \ref{backward trellis representation} shows the backward trellis, where the input to the decoder is at the right hand side of the decoder and its output is at the left hand side, which operates in a reverse direction of the conventional order.

\vspace{1mm}
\begin{figure}[!ht]
   \centering
  \includegraphics[width=2.2in,height=1.0in]{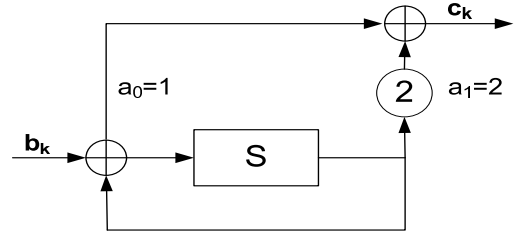}
  \caption{The encoder of code $g(x)=\frac{1+2x}{1+x}$.}\label{forward representation}
\end{figure}

\vspace{1mm}
\begin{figure}[!ht]
 \centering
  \includegraphics[width=1.7in,height=1.5in]{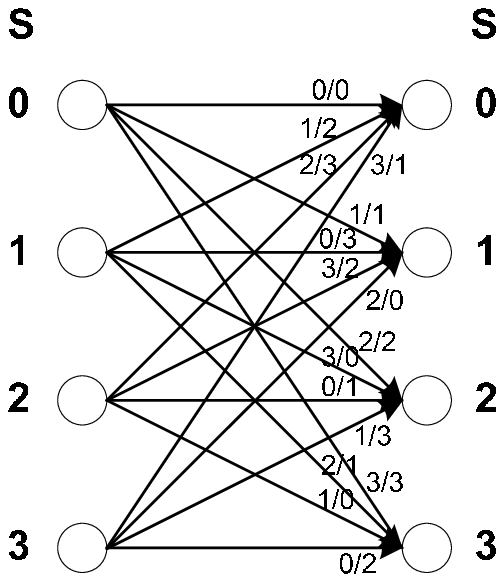}
  \caption{The trellis of code $g(x)=\frac{1+2x}{1+x}$.}\label{trellis representation}
\end{figure}

\vspace{1mm}
\begin{figure}[!ht]
 \centering
  \includegraphics[width=3.9in,height=1.6in]{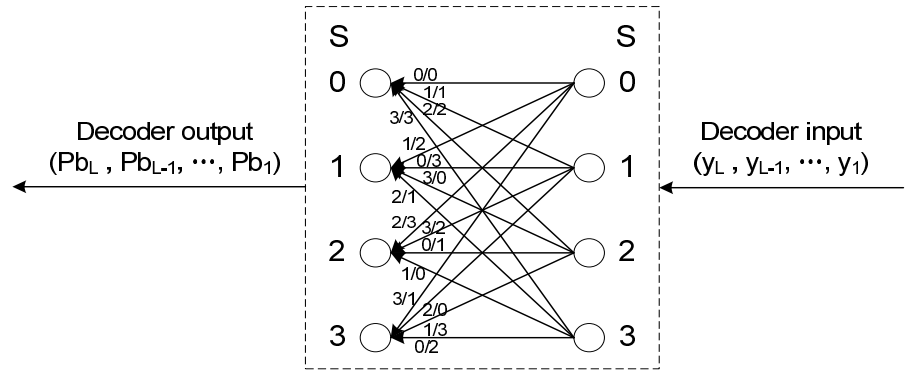}
  \caption{The backward trellis of code $g(x)=\frac{1+2x}{1+x}$.}\label{backward trellis representation}
\end{figure}

For ease of exposition, we propose to present the backward trellis in the forward direction where the decoder input and output are changed to the conventional order. Specifically, for a convolutional encoder, described by $g(x)=\frac{a(x)}{f(x)}=\frac{1+a_1x+\cdots+a_{n-1}x^{n-1}+x^{n}}{1+f_1x+\cdots+f_{n-1}x^{n-1}+x^{n}}$, if the labeling of the $k$th shift register in the encoder is changed from $S_k$ to $S_{n-k}$ and their respective coefficients are changed from from $a_k$ to $a_{n-k}, k=1, 2, \cdots, n$, and from $b_k$ to $b_{n-k}$, the resulting encoder is referred to as the reverse-memory labeling encoder of $g(x)$. For example, Fig. \ref{forward representation of the backward trellis} shows the forward representation of the backward trellis of code $g(x)=\frac{1+2x}{1+x}$. Its corresponding reverse-memory labeling encoder is shown in Fig. \ref{encoder of the backward trellis}.

\vspace{9mm}
\begin{figure}[!ht]
 \centering
  \includegraphics[width=3.9in,height=1.8in]{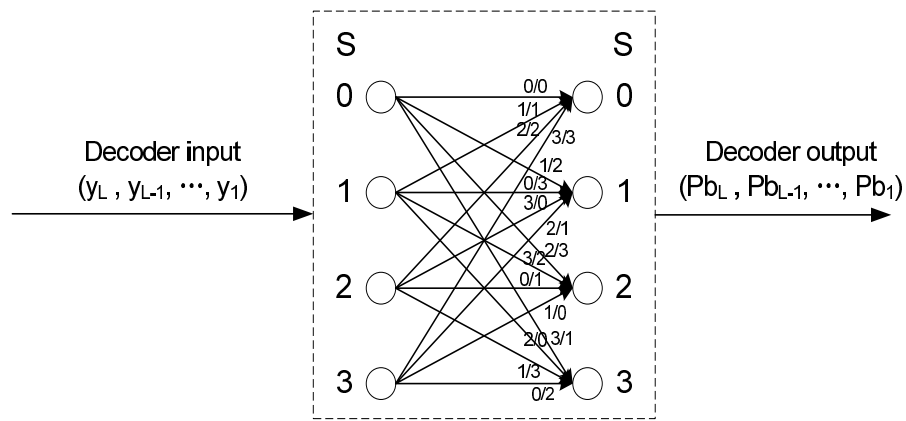}
  \caption{The equivalent forward representation of the backward trellis of code $g(x)=\frac{1+2x}{1+x}$.}\label{forward representation of the backward trellis}
\end{figure}

\vspace{5mm}
\begin{figure}[!ht]
 \centering
  \includegraphics[width=2.3in,height=1in]{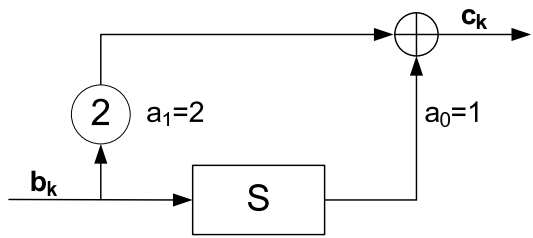}
  \caption{The encoder corresponds to the trellis of Fig. \ref{forward representation of the backward trellis}.}\label{encoder of the backward trellis}
\end{figure}

It is shown in \cite[Theorem 3]{YQSB13} that the relationship of the encoders for the forward and backward trellises can be extended to general rate-1 convolutional codes in $GF(q)$, as shown in the following theorem.
\begin{theorem}\label{theorem 2}
Given an encoder with generator polynomial $g(x)=\frac{a(x)}{f(x)}=\frac{1+a_1x+\cdots+a_{n-1}x^{n-1}+x^{n}}{1+f_1x+\cdots+f_{n-1}x^{n-1}+x^{n}}$, the forward representation of its backward trellis can be implemented by its reverse-memory labeling encoder of the same generator polynomial $g(x)$.
\end{theorem}

\begin{proof}
This can be proved similarly as the proof of Theorem 3 in \cite{YQSB13}, and we omit it here.
\end{proof}

From Theorem \ref{theorem 1}, we know that the SISO forward decoding of a given convolutional code, generated by $g(x)=\frac{a(x)}{f(x)}$, can be implemented by its dual encoder described by $q(x)=\frac{f(x)z(x)}{a(x)z(x)}$, where $z(x)$ is the degree $l$ minimum complementary polynomial of $a(x)$. Then according to Theorem \ref{theorem 2}, the SISO backward decoding of the convolutional code can be implemented by its reverse-memory labeling encoder of $q(x)$. By combining Theorems \ref{theorem 1} and \ref{theorem 2}, we can obtain the dual encoder for SISO MAP backward decoding, which is summarized in the following Theorem.

%backward decoding duality of convolutional codes
\begin{theorem}\label{theorem 3}
{\bf dual encoder for SISO MAP backward decoding:} We consider a convolutional code, generated by $g(x)=\frac{a(x)}{f(x)}=\frac{1+a_1x+\cdots+a_{n-1}x^{n-1}+x^{n}}{1+f_1x+\cdots+f_{n-1}x^{n-1}+x^{n}}$. Let $z(x)$ be the degree-$l$ minimum complementary polynomial of $a(x)$. Its SISO backward decoding can be implemented by its dual encoder, described by $q(x)=\frac{f(x)z(x)}{a(x)z(x)}=1+\frac{h_1x+\cdots+h_{n+l-1}x^{n+l-1}}{1+x^{n+l}}$, with reverse-memory labeling and time-reverse input, shown in Fig. \ref{backward dual encoder convolution}.
\end{theorem}

\begin{proof}
See Appendix \ref{proof of theorem 3}.
\end{proof}

\vspace{6mm}
\begin{figure}[!ht]
  % \centering
  \includegraphics[width=4.3in,height=1.5in]{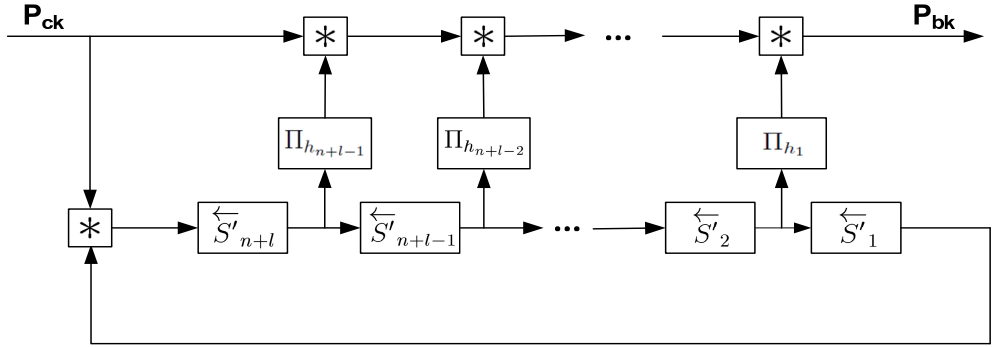}
  \caption{The dual encoder of the SISO backward decoding for the code $g(x)$, given by $q(x)=1+\frac{h_1x+\cdots+h_{n+l-1}x^{n+l-1}}{1+x^{n+l}}$.}\label{backward dual encoder convolution}
\end{figure}

The computational complexity of the dual encoder for the SISO backward decoding is dominated by the convolution operation. Similar to the forward decoding, we can apply FFT to further reduce the complexity of dual encoder for backward decoding. The FFT backward dual encoder is shown in Fig. \ref{backward dual encoder}.

\vspace{10mm}
\begin{figure}[!ht]
  % \centering
  \includegraphics[width=4.6in,height=2in]{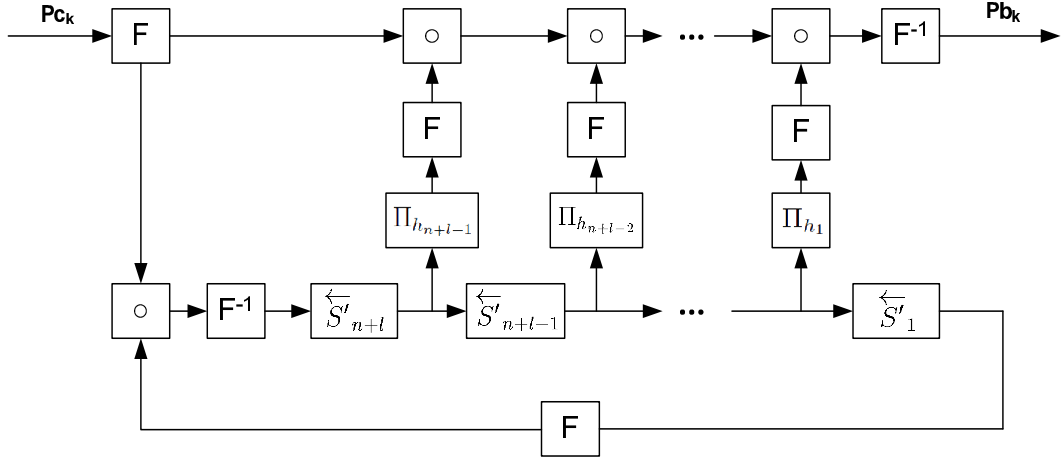}
  \caption{The FFT dual encoder of SISO backward decoding for the code $g(x)$.}\label{backward dual encoder}
\end{figure}

\section{The representation of bidirectional SISO MAP decoding}\label{shift register contents combining}
In the previous two sections, dual encoders for SISO MAP forward and backward decoding have been proposed. Based on the derived dual encoder structures, in this section, we represent the bidirectional SISO decoder by linearly combining shift register contents of the dual encoders for SISO MAP forward and backward decoding. We prove mathematically that such linear combining achieves exactly the same output as the bidirectional BCJR MAP decoding.

In the bidirectional BCJR MAP decoding, the APPs derived from the forward and backward recursions are combined at the same state at each time unit to obtain the desired decoding output. Therefore, it is usually assumed that the encoder begins with and ends at the all-zero state \cite{LC83}. The proposed dual encoder will produce the same output as the BCJR MAP algorithm when the forward and backward dual encoders have the same state at each time unit. As will be discussed shortly, this is ensured if the proposed dual encoder begins with and terminates at the all-zero state. To achieve this, tail symbols are added at the end of the code sequence.

Let us consider an encoder $\bar{C}$ of memory length $n+l$ in $GF(q)$, described by $q(x)=\frac{1}{g(x)}=\frac{f(x)z(x)}{a(x)z(x)}=1+\frac{h_1x+ \cdots + h_{n+l-1}x^{n+l-1}}{1+x^{n+l}}$. If the input to the encoder $\bar{C}$ is a codeword ${\overrightarrow{c}}=\left(c_1, c_2, \cdots, c_L\right)$, generated by $g(x)$, the output of the encoder $\bar{C}$ the decoded information sequence $\overrightarrow{b}$. Let us define $\left(c_{L+1},...,c_{L+n+l}\right)$ as the tail-bits required to terminate $\bar{C}$ at the all-zero state. Then following an analysis similar to that in \cite{YQSB13}, we can prove that the tail-biting convolutional encoder $\bar{C}$ has the following property.

\newtheorem{lemma}{\bf Lemma}
\begin{lemma}\label{dual decoder and encoder states returen to zero}
The tail-bits that terminate the encoder $\bar{C}$, described by $q(x)=1+\frac{h_1x+ \cdots + h_{n+l-1}x^{n+l-1}}{1+x^{n+l}}$, at the all-zero state also terminate the encoder C, generated by $g(x)=\frac{a(x)}{f(x)}$, at the all-zero state.
\end{lemma}

\begin{lemma}\label{same state transitions}
For a tail-biting convolutional encoder $\bar{C}$, generated by $q(x)$, and a given input sequence $(c_1, c_2, \cdots, c_L, c_{L+1}, \cdots, c_{L+n+l})$,  we define its backward encoder as the encoder of the same generator polynomial with reverse-memory labeling and time-reverse input $(c_{L+n+l}, \cdots , c_{L+1},$ $c_L, \cdots, c_2, c_1)$. Then the tail-biting encoder $\bar{C}$ and its backward encoder arrive at the same state at any time $k$.
\end{lemma}

In the decoding structures we introduced in the previous two sections, the input, output and shift register contents of dual encoders for forward and backward decoding are pmf vectors. To derive the bidirectional SISO decoder output, we need to combine the shift register contents of dual encoders for forward and backward decoding in an optimal way. Let ${\bf P_{S'_j(k)}}=[ p_{S'_j(k)}(0), p_{S'_j(k)}(1),$ $ p_{S'_j(k)}(\alpha), \cdots, p_{S'_j(k)}(\alpha^{q-2})]$ denote the combined pmf of the $j$th shift register of the combined dual encoder at time $k$. Since ${\bf P_{\overrightarrow{S'}_j(k)}}$ and ${\bf P_{\overleftarrow{S'}_j(k)}}$ are obtained from the forward decoding based on the received signals from time 1 to $k$ and that from backward decoding based on the received signals from time $L+n+l$ to $k+1$, they are independent. Furthermore, as shown in Lemma \ref{same state transitions}, for tail-biting encoder $\bar{C}$, generated by $q(x)$, forward and backward encoders will arrive at the same state at time $k$. Therefore, in the optimal combining, we have
\begin{align}\label{shift reister combining}
{\bf P_{S'_j(k)}}={\bf P_{\overrightarrow{S'}_j(k)}}\circ {\bf P_{\overleftarrow{S'}_j(k)}}.
\end{align}

Based on the dual encoder structure in Fig. \ref{dual encoder for forward decoding.}, the bidirectional SISO MAP decoding can be implemented by the proposed dual encoder with combined shift register contents. The output of the combined dual encoder is given by
\begin{align}\label{bidirectional decoder output}
{\bf P_{b_k}}={\bf P_{c_k}} \ast \Pi_{h_1} {\bf P_{S'_1(k-1)}} \ast \cdots \ast \Pi_{h_{n+l-1}} {\bf P_{S'_{n+l-1}(k-1)}}.
\end{align}

As shown in the following theorem, such  combining will produce exactly the same output as the bidirectional BCJR MAP algorithm.

\begin{theorem}\label{theorem 4}
We can represent the bidirectional SISO MAP decoder by linearly combining shift register contents of dual encoders for forward and backward decoding, as shown in (\ref{shift reister combining}) and (\ref{bidirectional decoder output}). This decoder produces exactly the same decoding output as the bidirectional BCJR MAP decoding algorithm.
\end{theorem}

\begin{proof}
See Appendix \ref{proof of theorem 4}.
\end{proof}

To reduce computational complexity, FFT can be applied to (\ref{bidirectional decoder output})
\begin{align}
{\bf P_{b_k}}=F^{-1}\left\{{F[\bf P_{c_k}}]  \circ F[\Pi_{h_1} {\bf P_{S'_1(k-1)}} ] \circ \cdots  \circ F[\Pi_{h_{n+l-1}} {\bf P_{S'_{n+l-1}(k-1)}}]\right\}.
\end{align}

Next, let us present some simulation results to validate our proposed scheme. A BPSK modulation is assumed. A frame size of $L=256$ symbols is employed over AWGN channels.

The bit error rate (BER) of various 4-state and 16-state convolutional codes are shown in Figs. \ref{BER performance of $G(D)=1+D$ FFC code} to \ref{BER performance of $132/112$ GC code}. The curve ``dual encoder forward+backward'' refers to the direct summation of the forward and backward dual encoder outputs, and the curve ``dual encoder shift register combined output'' refers to the optimal combined output (\ref{bidirectional decoder output}).

Figs. \ref{BER performance of $G(D)=1+D$ FFC code} to \ref{BER performance of $132/112$ GC code} show that the direct summation of the forward and backward dual encoder outputs suffers from some performance loss when compared to the bidirectional BCJR MAP algorithm. The SNR loss relative to the bidirectional BCJR MAP algorithm is 0, 0.1, 0.48, 0.1 and 1 $\mbox{dB}$ for codes $g(x)=1+x$, $g(x)=1+3x+2x^2$, $g(x)=1+x+2x^2$, $g(x)=\frac{1+x}{1+2x}$ and $g(x)=\frac{1+3x+2x^2}{1+x+2x^2}$. However, the proposed optimal linear combining scheme achieves exactly the same performance as the bidirectional BCJR MAP algorithm.

\section{Conclusions}\label{conclusion}
In this paper, we investigated the BCJR MAP decoding of rate-1 convolutional codes in $GF(q)$. We observed an explicit relationship between the SISO  BCJR MAP forward and backward decoder of a convolutional code and its encoder. Based on this observation, we proposed dual encoders for forward and backward decoding. The input of the dual encoders is the probability mass function of the code symbols and the output of the dual encoders is the probability mass function of the information symbols. The bidirectional SISO decoder is implemented by linearly combining the shift register contents of the dual encoders for forward and backward decoding. The proposed dual encoders significantly reduced the computational complexity of the bidirectional BCJR MAP decoding from exponential to linear in terms of convolutional code constraint length. To further reduce the complexity, fast Fourier transform is employed. Mathematical proofs and simulation results validate that the proposed dual encoder with shift register contents combining produces exactly the same output as the BCJR MAP decoding algorithm.

\appendices
\section{proof of theorem 1}\label{proof of theorem 1}
We consider the BCJR forward decoding algorithm of a general convolutional code $g(x)$ in $GF(q)$. Its dual encoder for forward decoding is described by $q(x)=1+\frac{h_1x+ \cdots + h_{n+l-1}x^{n+l-1}}{1+x^{n+l}}$. If the state of the dual encoder transits from $\left(u'_1, u'_2, \cdots, u'_{n+l}\right)$ at time $k-1$ to $\left(u_1, u_2, \cdots, u_{n+l}\right)$ at time $k$ with input $c_k$, then the probability of $b_k=\omega$ can be expressed as
\begin{align}\label{ak-1rk}
P_{b_k}(\omega)&=P\left\{b_k=\omega|\overrightarrow{y}\right\}=\sum_{(u', u)=U(b_k=\omega)}\alpha_{k-1}\left(u'\right)\gamma_k\left(u', u\right)\nonumber\\
&=\sum_{(u',u)=U(b_k=\omega)} \prod_{j=1}^{n+l} P_{\overrightarrow{S'}_j(k-1)}(u'_j) P(c_k)\nonumber\\
&=\sum_{u'_1, u'_2, \cdots, u'_{n+l}, \sum_{j=1}^{n+l-1} h_j u'_j +c_k=\omega}\prod_{j=1}^{n+l} P_{\overrightarrow{S'}_j(k-1)}(u'_j) P(c_k).
\end{align}
According to the generator polynomial $q(x)=1+\frac{h_1x+ \cdots + h_{n+l-1}x^{n+l-1}}{1+x^{n+l}}$, the dual encoder output is independent of the shift register contents in $\overrightarrow{S'}_{n+l}(k-1)$, shown in Fig. \ref{forward dual encoder}. Therefore (\ref{ak-1rk}) can be written as
\begin{align}\label{ak-1rk2}
P_{b_k}(\omega)
&=\sum_{u'_1, u'_2, \cdots, u'_{n+l}, \sum_{j=1}^{n+l-1} h_j u'_j+c_k=\omega}\prod_{j=1}^{n+l-1} P_{\overrightarrow{S'}_j(k-1)}(u'_j) P(c_k)\nonumber\\
&=\sum_{u'_1, u'_2, \cdots, u'_{n+l}, \sum_{j=1}^{n+l-1} h_j u'_j=\omega-c_k}\prod_{j=1}^{n+l-1} P_{\overrightarrow{S'}_j(k-1)}(u'_j) P(\omega-\sum_{j=1}^{n+l-1} h_j u'_j).
\end{align}
According to the definition of convolution operation, the probability mass function of $b_k$ can be written as
\begin{align}\label{puk}
{\bf P_{b_k}}={\bf P_{\sum_{j=1}^{n+l-1} h_j\overrightarrow{S'}_j(k-1)}} \ast  {\bf P_{c_k}}.
\end{align}

Using similar procedures of deriving (\ref{puk}) from (\ref{ak-1rk2}), we can get
\begin{align}
{\bf P_{b_k}}&={\bf P_{c_k}} \ast {\bf P_{h_1\overrightarrow{S'}_1(k-1)}} \ast \cdots \ast  {\bf P_{h_{n+l-1}\overrightarrow{S'}_{n+l-1}(k-1)}}\nonumber\\
&={\bf P_{c_k}} \ast \Pi_{h_1} {\bf P_{\overrightarrow{S'}_1(k-1)}} \ast \cdots \ast \Pi_{h_{n+l-1}} {\bf P_{\overrightarrow{S'}_{n+l-1}(k-1)}}.
\end{align}

This proves Theorem \ref{theorem 1}.

\section{proof of theorem 3}\label{proof of theorem 3}
We consider the backward decoding of convolutional codes in $GF(q)$. Let $\overleftarrow{S'}_j(k), j=1, 2, \cdots, n+l,$ denote the memory of the $j$-th shift register of the backward encoder of $\bar{C}$, generated by $q(x)$. Let ${\bf P_{\overleftarrow{S'}_j(k)}}=[ p_{\overleftarrow{S'}_j(k)}(0), p_{\overleftarrow{S'}_j(k)}(1), \cdots, p_{\overleftarrow{S'}_j(k)}(q-1)]$ denote the pmf vector of $\overleftarrow{S'}_j(k)$. The probability that $b_k=\omega$ is given by
\begin{align}
P_{b_k}(\omega)&=P\left\{b_k=\omega|\overrightarrow{y}\right\}=\sum_{(u',u)=U(b_k=\omega)}\beta_{k}\left(u\right)\gamma_k\left(u', u\right)\nonumber\\  \label{backward all shift register contents}
&=\sum_{(u',u)=U(b_k=\omega)} \prod_{i=1}^{n+l} P_{\overleftarrow{S'}_i(k)}(u_i) P(c_k)\\ \label{backward no first shift register}
&=\sum_{(u',u)=U(b_k=\omega)} \prod_{i=2}^{n+l} P_{\overleftarrow{S'}_i(k)}(u_i) P(c_k)\\ \label{bkrk}
&=\sum_{u'_1, u'_2, \cdots, u'_{n+l}, \sum_{j=1}^{n+l-1} h_j u'_j+c_k=\omega}\prod_{j=1}^{n+l-1} P_{\overleftarrow{S'}_j(k-1)}(u'_j) P(c_k).
\end{align}

Note that (\ref{backward no first shift register}) is derived from (\ref{backward all shift register contents}) because at time slot $k$, the dual encoder output for BCJR MAP backward decoding is independent of $P_{\overleftarrow{S'}_1(k)}$. Based on (\ref{bkrk}), we can get
\begin{align}
{\bf P_{b_k}}&={\bf P_{c_k}} \ast {\bf P_{h_1\overleftarrow{S'}_1(k-1)}} \ast \cdots \ast  {\bf P_{h_{n+l-1}\overleftarrow{S'}_{n+l-1}(k-1)}}\nonumber\\
&={\bf P_{c_k}} \ast \Pi_{h_1} {\bf P_{\overleftarrow{S'}_1(k-1)}} \ast \cdots \ast \Pi_{h_{n+l-1}} {\bf P_{\overleftarrow{S'}_{n+l-1}(k-1)}}.
\end{align}

This proves Theorem \ref{theorem 3}.

\section{Proof of theorem 4}\label{proof of theorem 4}
We consider a convolutional code in $GF(q)$, generated by $g(x)$. Its dual encoder for decoding is described by $q(x)$. It is assumed that the state of the dual encoder $\bar{C}$ transits from $\left(u'_1, u'_2, \cdots, u'_{n+l}\right)$ at time $k-1$ to $\left(u_1, u_2, \cdots, u_{n+l}\right)$ at time $k$ with input $c_k$. For the bidirectional BCJR MAP algorithm, the probability that $b_k=\omega$ is given by
\begin{align}\label{bidirectional decoder proof}
P_{b_k}(\omega)&=P\left\{b_k=\omega|\overrightarrow{y}\right\}=\sum_{(u', u)=U(b_k=\omega)}\alpha_{k-1}\left(u'\right)\gamma_k\left(u', u\right)\beta_k\left(u\right)\nonumber\\
&=\sum_{(u',u)=U(b_k=\omega)} \prod_{j=1}^{n+l} P_{\overrightarrow{S'}_j(k-1)}(u'_j) P(c_k)\prod_{i=1}^{n+l} P_{\overleftarrow{S'}_j(k)}(u_i)\nonumber\\
&=\sum_{(u',u)=U(b_k=\omega)} \prod_{j=1}^{n+l-1} P_{\overrightarrow{S'}_j(k-1)}(u'_j) P(c_k)\prod_{i=2}^{n+l} P_{\overleftarrow{S'}_j(k)}(u_i)\nonumber\\
&=\sum_{u'_1, u'_2, \cdots, u'_{n+l}, \sum_{j=1}^{n+l-1} h_j u'_j +c_k=\omega}\prod_{j=1}^{n+l-1} P_{\overrightarrow{S'}_j(k-1)}(u'_j) P(c_k)\prod_{j=1}^{n+l-1} P_{\overleftarrow{S'}_j(k-1)}(u'_j)\nonumber\\
&=\sum_{u'_1, u'_2, \cdots, u'_{n+l}, \sum_{j=1}^{n+l-1} h_j u'_j+c_k=\omega}\prod_{j=1}^{n+l-1} P_{S'_j(k-1)}(u'_j) P(c_k),
\end{align}
where $P_{S'_j(k-1)}(u'_j) = P_{\overrightarrow{S'}_j(k-1)}(u'_j)P_{\overleftarrow{S'}_j(k-1)}(u'_j)$. From (\ref{bidirectional decoder proof}), we can get
\begin{align}\label{bidirectional decoder output BCJR}
{\bf P_{b_k}}&={\bf P_{c_k}} \ast {\bf P_{h_1S'_1(k-1)}} \ast \cdots \ast  {\bf P_{h_{n+l-1}S'_{n+l-1}(k-1)}}\nonumber\\
&={\bf P_{c_k}} \ast \Pi_{h_1} {\bf P_{S'_1(k-1)}} \ast \cdots \ast \Pi_{h_{n+l-1}} {\bf P_{S'_{n+l-1}(k-1)}}.
\end{align}

Comparing the shift register combined outputs of the dual encoder in (\ref{bidirectional decoder output}) and the outputs of the bidirectional BCJR MAP algorithm in (\ref{bidirectional decoder output BCJR}), we can see that they are exactly of the same.

This proves Theorem \ref{theorem 4}.

\bibliographystyle{IEEEtran}
\bibliography{IEEEabrv,dualdecoder}

\newpage

\begin{figure}[!h]
  % \centering
  \includegraphics[width=6.3in]{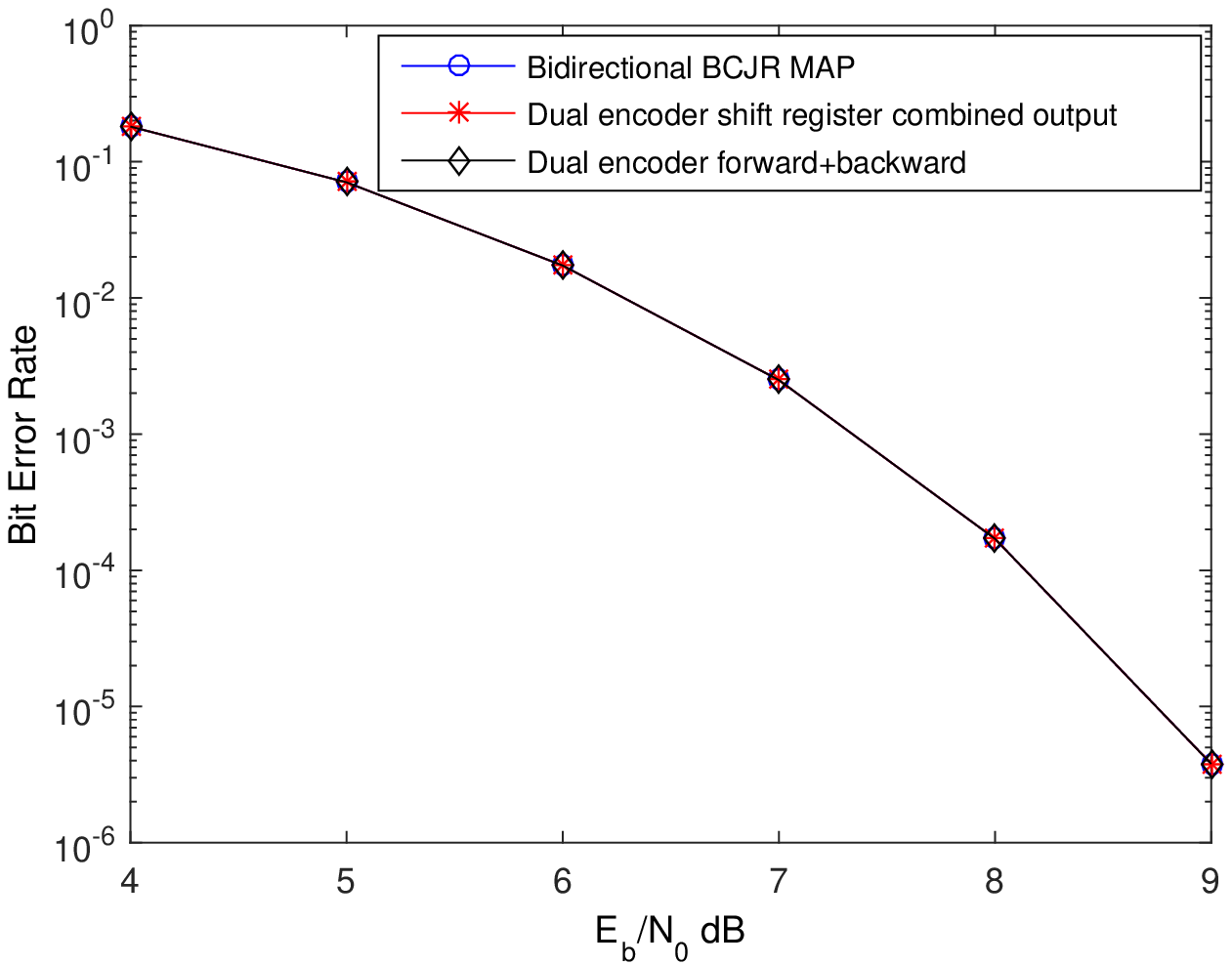}
  \caption{BER performance of $g(x)=1+x$ code over AWGN channels.}\label{BER performance of $G(D)=1+D$ FFC code}
\end{figure}

\begin{figure}[!h]
  % \centering
  \includegraphics[width=6.3in]{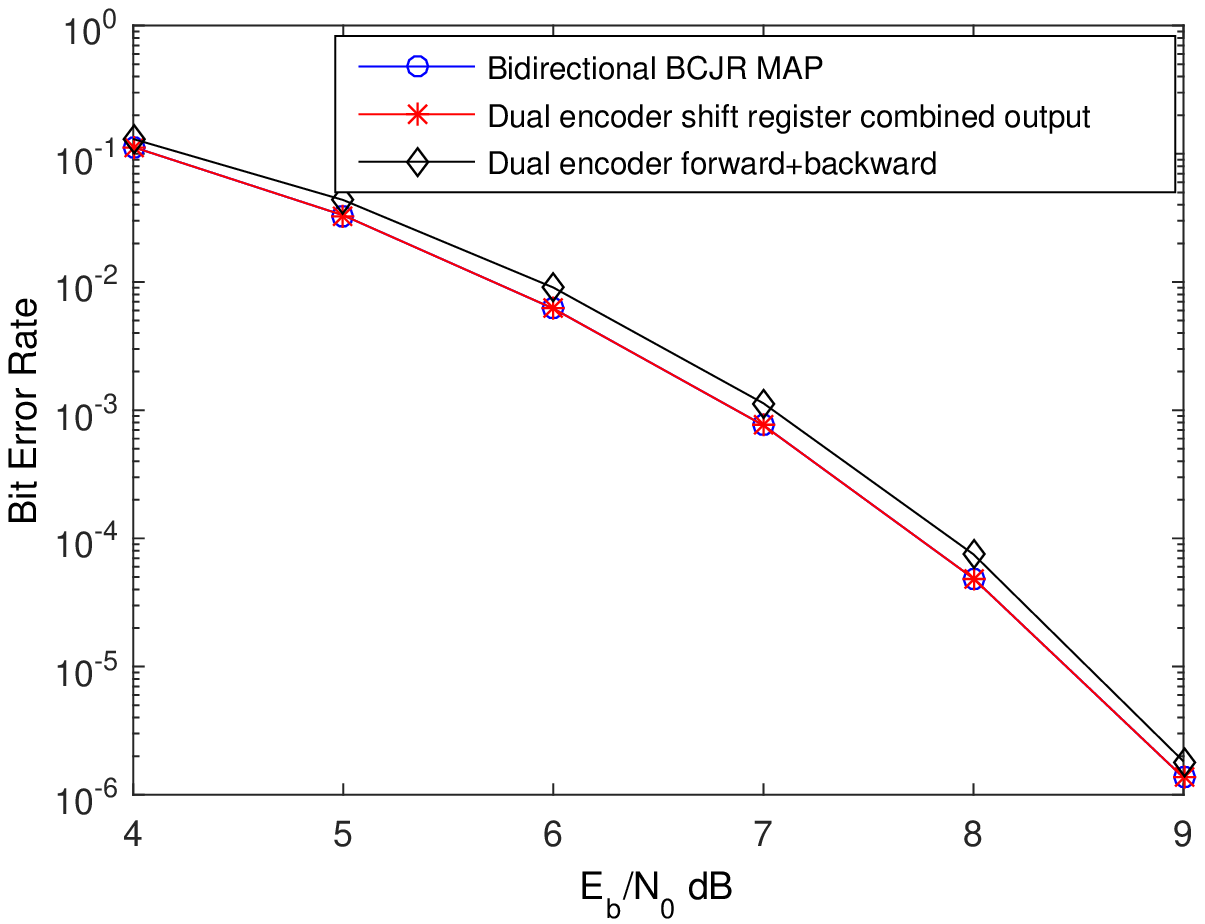}
  \caption{BER performance of $g(x)=1+3x+2x^2$ code over AWGN channels.}\label{BER performance of $G(D)=1+3D+2D^2$ FFC code}
\end{figure}

\begin{figure}[!h]
  % \centering
  \includegraphics[width=6.3in]{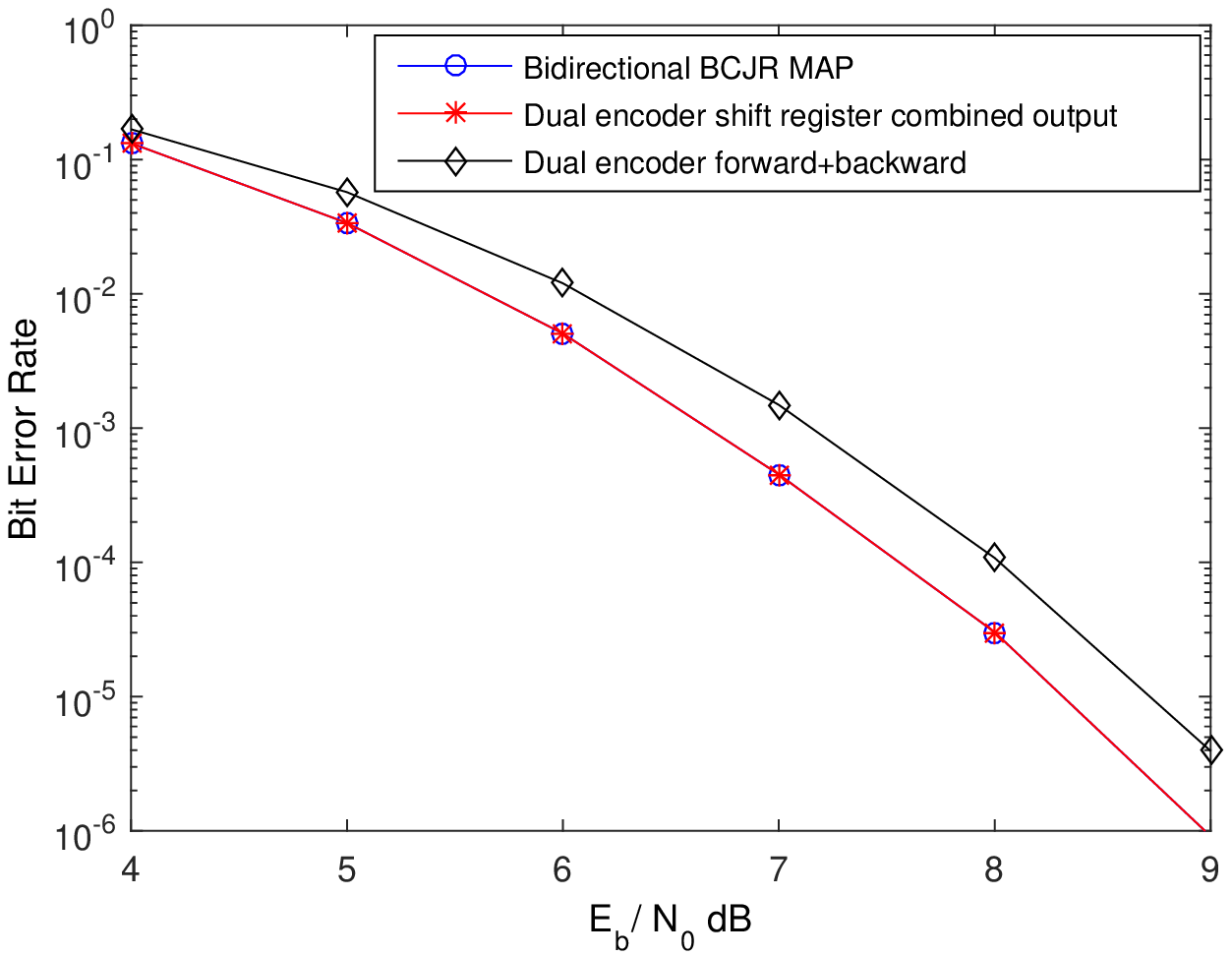}
  \caption{BER performance of $g(x)=1+x+2x^2$ code over AWGN channels.}\label{BER performance of $G(D)=1+D+2D^2$ FFC code}
\end{figure}

\begin{figure}[!h]
  % \centering
  \includegraphics[width=6.3in]{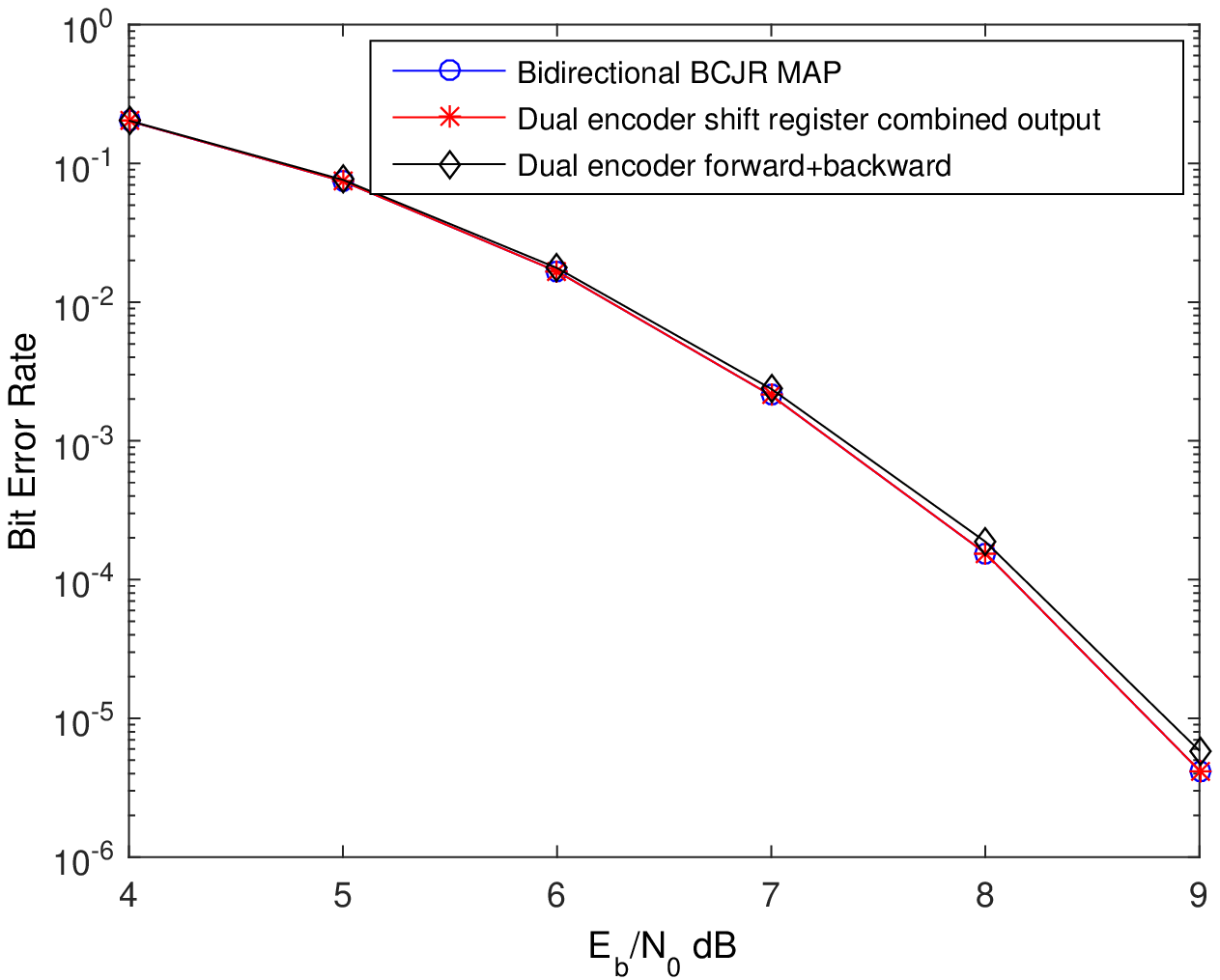}
      \caption{BER performance of $g(x)=\frac{1+x}{1+2x}$ code over AWGN channels.} \label{BER performance of $132/112$ GC code}
\end{figure}

\begin{figure}[!h]
  % \centering
  \includegraphics[width=6.3in]{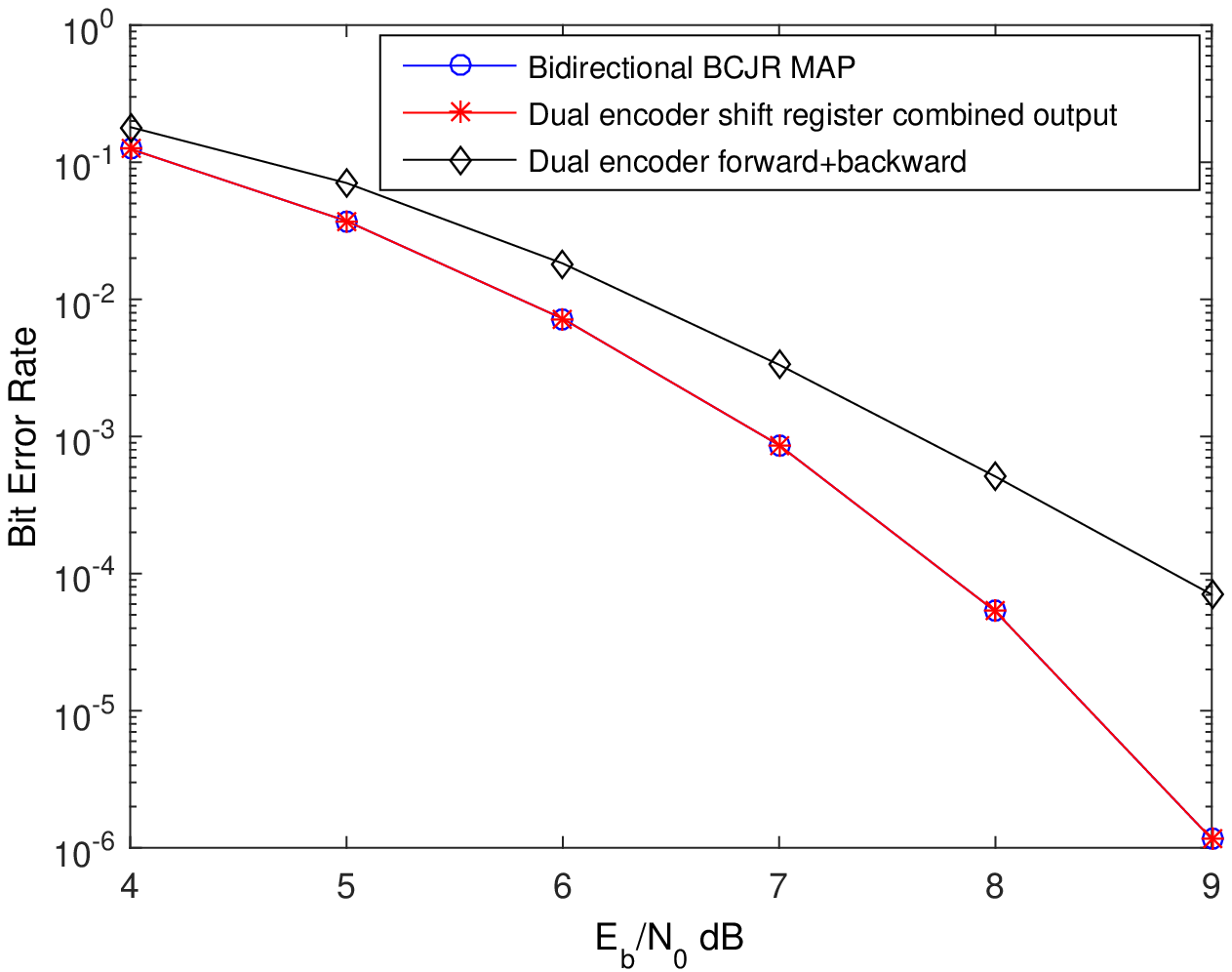}
      \caption{BER performance of $g(x)=\frac{1+3x+2x^2}{1+x+2x^2}$ code over AWGN channels.} \label{BER performance of $132/112$ GC code}
\end{figure}

\end{document}